\documentclass[10pt]{article}
\pdfoutput=1
\usepackage{graphicx,scalerel}
\usepackage{epic,graphpap,graphics,float,tabularx,comment}
\usepackage{paralist,longtable,fancyhdr,fancyvrb,bbm,natbib,xcolor}
\usepackage[hyperfootnotes=false]{hyperref}
\usepackage{amsmath}
\usepackage{amstext}
\usepackage{amsbsy}
\usepackage{amsthm}
\usepackage{amsfonts}
\usepackage{amssymb}
\usepackage{amscd}
\usepackage{eucal} 
\usepackage{eurosym}

\theoremstyle{plain}
\newtheorem{thm}{Theorem}[section]
\newtheorem{lem}[thm]{Lemma}
\newtheorem{cor}[thm]{Corollary}
\newtheorem{prop}[thm]{Proposition}
\theoremstyle{definition}
\newtheorem{defn}[thm]{Definition}
\newtheorem{exa}[thm]{Example}
\newtheorem{rem}[thm]{Remark}
\newtheorem{pb}[thm]{Problem}
\numberwithin{equation}{section}
\numberwithin{equation}{section}

\def\sumi{\sum_{i=1}^n}
\def\sumz{\sum_{i=0}^n}
\def\sumj{\sum_{j=1}^n}

\def\a{\alpha}
\def\m{\mu}
\def\d{\delta}
\def\t{\tau}
\def\p{\pi}

\def\zh{{\hat z}}
\def\fh{{\hat f}}

\def\g{\gamma}

\def\s{\sigma}

\def\D{\Delta}
\def\F{{\mathcal F}}

\def\S{{\mathbf S}}
\def\DD{{\mathbf D}}

\def\SS{{\widetilde{\mathbf S}}}

\def\V{{\mathbf V}}
\def\Vt{{\mathbf{\widetilde V}}}
\def\Vh{{\mathbf{\widehat V}}}

\def\U{{\widetilde U}}

\def\ph{\varphi}

\def\1{{\mathbbm 1}}

\def\N{{\mathbb N}}

\def\R{{\mathbb R}}

\def\eqdef{\triangleq}

\def\intt{\int_0^t}

\def\sumi{\sum_{i=1}^n}
\def\sumij{\sum_{i,j=1}^n}

\def\half{\frac{1}{2}}
\def\brac#1{\langle #1\rangle}
\def\bbrac#1{\big\langle #1 \big\rangle}

\def\ito{It\^o}

\def\intt{\int_0^t}

\def\dd#1#2{\frac{\partial #1}{\partial #2}}

\allowdisplaybreaks

\begin{document}

\centerline{\bf \LARGE{Portfolios generated by contingent claim functions,}}
\vspace{10pt}
\centerline{\bf \LARGE{with applications to option pricing}}

\vskip 35pt

\centerline{\large{Ricardo T. Fernholz\footnote[1]{Claremont McKenna College, 500 E. Ninth St., Claremont, CA 91711, rfernholz@cmc.edu.} \hskip 60pt Robert Fernholz\footnote[2]{Allocation Strategies, LLC, Princeton, NJ, bob@bobfernholz.com.}}}

\vskip 35pt

\centerline{\large{\today}}

 \vspace{-15pt}

\vskip 60pt

\begin{abstract}
This paper presents a synthesis of the theories of portfolio generating functions and option pricing. The theory of portfolio generation is extended to measure the value of portfolios generated by positive $C^{2,1}$  functions of asset prices $X_1,\ldots,X_n$ directly, rather than with respect to a numeraire portfolio. If a portfolio generating function satisfies a specific partial differential equation, then the value of the portfolio generated by that function will replicate the value of the function. This differential equation is a general form of the Black-Scholes equation. Similar results apply to contingent claim functions, which are portfolio generating functions that are homogeneous of degree~1. With the addition of a riskless asset, an inhomogeneous portfolio generating function $\V\colon\R^{+n}\times[0,T]\to\R^+$ can be extended to an equivalent contingent claim function $\Vh\colon\R^+\times\R^{+n}\times[0,T]\to\R^+$ that generates the same portfolio and is replicable if and only if $\V$ is replicable. Several examples are presented.
\end{abstract}

\vskip 240pt


\vskip 5pt


\vfill

\section{Introduction} \label{intro} 

Functionally generated portfolios were introduced in \citet{F:pgf,F:2002} as one of the central concepts of stochastic portfolio theory. Option-pricing theory began with the seminal paper of \citet{Black/Scholes}, who used stochastic calculus to price European options on financial assets. Here we present a synthesis of portfolio generation and option pricing. We show that portfolio generation has a natural application to option pricing, while option pricing can be considered an early example of portfolio generation (we note that the remark before Problem~3.2.4 of \citet{F:2002} suggested such a relationship).

Functionally generated portfolios were originally defined for stock markets, and measured the value of a portfolio versus the market portfolio as numeraire \citep{F:pgf}. Later, functional generation with more general numeraire portfolios was developed \citep{Strong:2014}, and portfolios generated by functions of ranked market weights appeared in \citet{F:rank}. Functionally generated portfolios were originally defined multiplicatively, and later the theory was extended to additive portfolio generation \citep{Karatzas/Ruf:2017}. Portfolios generated by functions that depend on the sample paths of the assets were introduced by \citet{SSV:2018} and \citet{KKi:2020}, and further extended to path-functional portfolios by \citet{Cuchiero-Moller:2024}. Portfolio generating functions have been used for a number of different purposes, from explaining the size effect for stocks \citep{Banz:1981,F:2002} to proving the existence of arbitrage in realistic models of stock markets (see, e.g., \citet{Fernholz/Karatzas/Kardaras:2005}, \citet{FKR:2018}, or \citet{KK:2021}). 

The modern theory of option pricing started with \citet{Black/Scholes}, who derived the well-known Black-Scholes differential equation and then solved this equation for the price of European put and call options. This theory was advanced by \citet{Merton:1973}, who developed the concept of a dynamically rebalanced portfolio to replicate the prices of options following the Black-Scholes model (see, e.g., the discussion of ``perfect substitutability" and ``perfect hedge" on page~169 of \citet{Merton:1973}). The early theory of option pricing is reviewed by \citet{Smith:1976}, and option-pricing theory has continued to develop with multivariate versions of the Black-Scholes model (see, e.g., \citet{Carmona2006} and \citet{Guillaume2019}) as well as results for wider classes of options and price processes (see, e.g., \citet{HN:2001,HN2:2001} and   \citet{RSS:2007}).

\citet{Merton:1973} recognized that a natural property of option-pricing functions is that they are {\em homogeneous of degree~1}. This homogeneity requires that a change of numeraire currency affects both asset prices and the value of a contingent claim in the same way. For example, if we switch from measuring asset prices in euros rather than dollars, then the same exchange rate should be applied to the value of a contingent claim as to the prices of the underlying assets. \citet{HN:2001,HN2:2001} argue that scale invariance is a fundamental property of any well-posed market model, and this leads naturally to homogeneity of degree 1 for all contingent claims. 

Here we use multiplicative portfolio generation by functions of asset prices rather than market weights, and we measure the value of portfolios directly, with no numeraire portfolio. We show that a positive $C^{2,1}$ function of asset prices and time will generate a portfolio and that the value of this portfolio can be decomposed into the generating function plus a drift process, in a manner similar to portfolio value decomposition with the market as numeraire \citep{F:2002}. 

We establish necessary and sufficient conditions for a portfolio generating function to be \emph{replicable}, so that the value of the portfolio generated by that function replicates the value of the function. These conditions require that the generating function satisfy a partial differential equation that is a form of the Black-Scholes equation. We apply these results to contingent claim functions, defined to be portfolio generating functions that are homogeneous of degree~1. Because any contingent claim function that satisfies this differential equation is replicable, it follows that any such function is equal to the value of an option with expiration payout equal to the terminal value of that function. Furthermore, our results establish that the replicating portfolio for this option is the portfolio generated by the contingent claim function. In this way we show that central concepts of two branches of mathematical finance --- stochastic portfolio theory and option-pricing theory --- are closely related. 

Although homogeneity may have ample economic justification in the context of option pricing, there is no mathematical need for it in portfolio generation. Accordingly, we show that a positive $C^{2,1}$ function of the asset prices and time will generate a portfolio whether or not it is homogeneous of degree~1. We establish a three-step procedure to convert an inhomogeneous portfolio generating function $\V\colon\R^{+n}\times[0,T]\to\R^+$  into an equivalent contingent claim function $\Vh\colon\R^+\times\R^{+n}\times[0,T]\to\R^+$ that generates the same portfolio and is replicable if and only if $\V$ is replicable. We apply this procedure to a model with a single asset with constant parameters, and we show that the general differential equation for a replicable generating function is identical to the original Black-Scholes equation, and the solution to this general differential equation will be the same as in \citet{Black/Scholes}. Our results continue to hold in the multivariate settings of  \citet{HN:2001,HN2:2001}, \citet{Carmona2006}, and \citet{Guillaume2019}.

\section{Setup} \label{setup} 
 
We consider  $n\ge1$ assets  represented by  {\em price processes} $X_1,\ldots,X_n$, which are strictly positive continuous semimartingales $X_i$ defined on $[0,T]$, for $T>0$, that satisfy 
\begin{equation}\label{1.0}
d\log X_i(t) = \g_i(t) dt+\sum^d_{\ell=1}\zeta_{i,\ell}(t) dW_\ell(t),
\end{equation}
for $i=1,\ldots,n$, $t\in[0,T]$, and $d\ge n$, where $W=(W_1,\ldots,W_d)$  is a Brownian motion, the processes $\g_i$ are integrable, the processes $\zeta_{i,\ell}$ are square integrable, and both the $\g_i$ and $\zeta_{i,\ell}$ are adapted with respect  to the Brownian filtration $\F^W$.   The $\g_i$ are called {\em growth rate processes}. To simplify notation, we adopt the convention that all equations involving random variables or processes are with probability one under the relevant probability measure.   For a market to be {\em complete}, it is necessary that $n = d$ in~\eqref{1.0}, in which case certain hedging properties are known to hold (see \citet{Karatzas/Shreve:1998}, Section~1.6, or \citet{KK:2021}, Section~3.3). Here we are not interested in market completeness, and except in examples we shall not specify whether $d>n$ or $d=n$ in~\eqref{1.0}.  

The price processes $X_i$  are valued in terms of a fixed numeraire currency with value identically 1 at all times, so $X_i(t)$ represents the price of the $i$th asset at time $t$ in units of this currency. The numeraire currency serves only as a unit of account and is not a tradable asset.  A {\em riskless asset}  is a price process that is of finite variation, in which case \eqref{1.0} becomes
\begin{equation}\label{1.1a}
d\log X_0(t) = \g_0(t)dt,
\end{equation}
for $t\in[0,T]$, and  $\g_0$ is called the {\em interest rate process} for $X_0$. Unless otherwise stated, we assume neither the existence nor the non-existence of one or more riskless assets. Since the numeraire currency serves only as a unit of account, the closest one can come to it in a portfolio is through the use of a riskless asset, if there is one.  

We represent the cross variation process for $\log X_i$ and $\log X_j$ by $\brac{\log X_i,\log X_j}_t$, and define the {\em covariance processes}  $\s_{ij}$ such that
\[
\s_{ij}(t)\eqdef\sum_{\ell=1}^d \zeta_{i,\ell}(t)\zeta_{j,\ell}(t)=\frac{d\brac{\log X_i,\log X_j}_t}{dt},
\]
for $i,j=1,\ldots,n$ and $t\in[0,T]$. For a riskless asset $X_0$, the covariance processes satisfy $\s_{i0}=\s_{0j}\equiv0$, for $i,j=0,1,\ldots,n$. The notation $\s^2_i$ may sometimes be used for $\s_{ii}$. Unless otherwise stated, we make no assumptions  regarding the rank of the matrix $\big(\s_{ij}(t)\big)_{\{i,j=1,\ldots,n\}}$ or regarding the existence of arbitrage. We sometimes use the vector notation $X=(X_1,\ldots,X_n)$, or with a riskless asset $X_0$, the notation $(X_0,X)=(X_0,X_1,\ldots,X_n)$.

A portfolio $\p$ of price processes $X_1,\ldots,X_n$ and riskless asset $X_0$ is identified by its {\em weight processes}, or {\em weights,} $\p_1,\ldots,\p_n$ and $\p_0=1-\p_1-\cdots-\p_n$, which are bounded  $\F^W$\!\!-adapted processes. For a portfolio $\p$ and price processes $X_0,X_1,\ldots,X_n$, the {\em portfolio value process} $Z_\p$ will satisfy
\begin{equation}\label{1.1.1}
dZ_\p(t) \eqdef Z_\p(t)\sumz\p_i(t)\frac{dX_i(t)}{X_i(t)},
\end{equation}
for $t\in[0,T]$. This proportional definition of the portfolio value process ensures that the portfolio is {\em self-financing,} since it directly measures the change in portfolio value induced proportionally by the changes in the prices of the component assets (see \citet{F:2002}, equation~(1.1.12), or \citet{KK:2021}, Section~1.3).
 
In logarithmic form, \ito's rule implies that~\eqref{1.1.1} becomes
\begin{equation}\label{1.2}
d\log Z_\p(t)=\sumz\p_i(t)\,d\log X_i(t)+\g^*_\p(t)dt,
\end{equation}
for $t\in[0,T]$, with the {\em excess growth rate process}
\begin{equation}\label{1.3}
\g^*_\p(t)\eqdef\half\bigg(\sumi\p_i(t)\s_{ii}(t)-\sumij\p_i(t)\p_j(t)\s_{ij}(t)\bigg),
\end{equation}
for $t\in[0,T]$. Note that the summation here starts with $i=1$ since $\s_{i0}=\s_{0j}\equiv0$. Details regarding this logarithmic representation for portfolios can be found in \citet{F:2002}, Section~1.1. 

We use partial differentiation of $C^{2,1}$ functions of the form $f\colon \R^{+n}\times[0,T]\to \R^+$, and we use the notation $D_i f$, for $i=1,\ldots,n$, for the partial derivative of $f$ with respect to the $i$th argument of $f$, and $D_t f$ for the partial derivative of $f$ with respect to the last argument, which represents time $t\in[0,T]$. In order to avoid confusion when we consider derivatives of composed functions $f\circ g$, we sometimes use classical notation for partial derivatives. For example, for $C^1$ functions $f\colon\R^{+n}\to\R^+$ and $g=(g_1,\ldots,g_n)\colon\R^+\hspace{0pt}^m\to\R^{+n}$ the partial derivative $D_i(f\circ g)(x)$ with respect to $x_i$, for $x=(x_1,\ldots,x_m)\in\R^+\hspace{0pt}^m$, will be written with the notation
\begin{equation}\label{1.6}
\dd{}{x_i}\big(f(g(x)\big)= D_i(f\circ g)(x) =\sumj D_j f(g(x))D_ig_j(x),
\end{equation}
for $i=1,\ldots,m$, by the chain rule. Finally, we use the Kronecker delta notation, $\d_{ij}\eqdef\1_{i=j}$, for $i,j\in\N$.

\section{Portfolio generating functions}\label{generating}

Portfolio generating functions first appeared in \citet{F:pgf} as positive continuous functions defined on the open unit simplex $\D^n\eqdef\{x\in\R^{+n}\colon x_1+\cdots+x_n=1\}$. In that setting, $X_1,\ldots,X_n$ represented the capitalizations of companies rather than prices, and they were defined to follow the same logarithmic representation~\eqref{1.0}. There was no riskless asset, and the $X_i$  comprised a {\em market} $\{X_1,\ldots,X_n\}$, with the {\em market capitalization} represented by
\begin{equation}\label{1.5}
Z(t)\eqdef X_1(t)+\cdots+X_n(t),
\end{equation}
for  $t\in[0,T]$. The {\em market portfolio} $\m$ was defined by its weight processes 
\begin{equation}\label{1.4}
\m_i(t)\eqdef X_i(t)/Z(t),
\end{equation}
for $i=1,\ldots,n$ and $t\in[0,T]$, so $\m(t)\in\D^n$,  for $t\in[0,T]$, and its value process $Z_\m$, defined as in~\eqref{1.1.1}, was shown to satisfy 
\[
d\log Z_\m(t)=d\log Z(t),
\]
 for $t\in[0,T]$. 
 
Time dependent portfolio generating functions were introduced in \citet{F:2002}, Section~3.2, in which a continuous function $\S\colon\D^n\times[0,T]\to\R^+$ was defined to  generate a portfolio $\p$ if
\begin{equation}\label{3.10}
d\log\big(Z_\p(t)/Z_\m(t)\big)=d\log\S(\m(t),t)+d\Theta_\S(t),
\end{equation}
for $t\in[0,T]$, where $\Theta_\S$ is a process of finite variation with $\Theta_\S(0)=0$.  The process $\Theta_\S$ was called the {\em drift process} for $\S$. It was shown that if $\S\colon U\times[0,T]\to\R^+$ is $C^{2,1}$ for an open neighborhood $U$ of $\D^n\subset\R^n$, with $x_i D_i\log\S(x,t)$ bounded on $\D^n\times[0,T]$, for $i=1,\ldots,n$,  then $\S$ generates the portfolio $\p$ with weights
\begin{equation}\label{3.9}
\p_i(t)=\Big(D_i\log\S(\m(t),t)+1-\sumj\m_j(t)D_j\log\S(\m(t),t)\Big)\m_i(t),
\end{equation}
for $i=1\ldots,n$ and $t\in[0,T]$, and $\Theta_\S$ given by
\begin{equation}\label{3.11}
d\Theta_\S(t)=-\half\sumij\frac{D_{ij}\S(\m(t),t)}{\S(\m(t),t)}\m_i(t)\m_j(t)\t_{ij}(t)dt-\frac{D_t\S(\m(t),t)}{\S(\m(t),t)}dt,
\end{equation}
for $t\in[0,T]$, where $\t_{ij}(t)dt=d\brac{\log\m_i,\log\m_j}_t$ (see  \citet{F:2002}, Section~3.2).  

In the setting of~\eqref{3.10}, the market portfolio $\m$ served as the {\em numeraire portfolio;} more recently, \citet{Strong:2014} generalized this definition to the case where  an arbitrary portfolio $\nu$ replaced the market portfolio $\m$ as numeraire.  Here we define portfolio generating functions in terms of prices rather than market weights and we measure the value of portfolios directly, with no numeraire portfolio.

\begin{defn}\label{D3} Let $\V\colon\R^{+n}\times[0,T]\to\R^+$ be a continuous function, let $X_1,\ldots,X_n$ be price processes, and let $X_0$ be a riskless asset. Then $\V$ {\em generates} the portfolio $\p$ for  $X_0,X_1,\ldots,X_n$ if
\begin{equation}\label{3.4}
d\log Z_\p(t)  = d\log \V(X(t),t)  + d\Phi_\V(t),
\end{equation}
 for $t\in[0,T)$, where $\Phi_\V$ is a process of finite variation with $\Phi_\V(0)=0$.  If $\p_0\equiv0$, then $\V$ generates $\p$ for $X_1,\ldots,X_n$. The process $\Phi_\V$ is called the {\em drift process} for $\V$. 
 \end{defn}
 
This is the definition for {\em multiplicative} portfolio generation by functions of asset prices, and this requires that portfolio values be strictly positive. {\em Additive} portfolio generation was introduced in \citet{Karatzas/Ruf:2017}, in which case trading strategies can assume negative values. Also {\em path-dependent} portfolio generation was proposed in  \citet{SSV:2018} and \citet{KKi:2020}, and this was further extended by \citet{Cuchiero-Moller:2024} to path-functional portfolios. All the results we present here might be similarly extendable.

We shall consider portfolios generated by differentiable functions of the asset prices, but portfolios can be generated by non-differentiable functions as well. Here we present a simple example of  such a function (see also  \citet{F:rank}).

\begin{exa}\label{LT} {\em A portfolio generated by a non-differentiable function.} Let 
\[
\V(x_1,x_2,t)=x_1\lor x_2,
\]
for $(x_1,x_2,t)\in\R^{+2}\times[0,T]$. Then $\V$ generates the portfolio $\p$ for  $X_1,X_2$ with weights 
\[
\p_1(t)=\1_{X_1\ge X_2}(t)\quad\text{ and }\quad \p_2(t)=1-\p_1(t),
\]
for  $t\in[0,T]$, and drift process
\[
\Phi_\V(t)=-\Lambda_{\log (X_1/X_2)}(t), 
\]
for $t\in[0,T]$, where $\Lambda_{\log (X_1/X_2)}$ is the semimartingale local time at 0 for $\log (X_1/X_2)$ (see \citet{F:2002}, Lemma~4.1.10, regarding this example and  \citet{Karatzas/Shreve:1991}, Section~3.7, regarding semimartingale local time).  \qed
\end{exa}

\begin{defn}\label{PG1} A {\em portfolio generating function} is a continuous function $\V\colon\R^{+n}\times [0,T]\to\R^+$ that is $C^{2,1}$ on $\R^{+n}\times [0,T)$ such that the functions $x_i D_i \V(x,t)/\V(x,t)$, for $i=1,\ldots,n$,  are bounded for $(x,t)\in\R^{+n}\times[0,T)$. The {\em terminal value function} for the portfolio generating function  $\V$ is the function $\V(\,\cdot\,,T)\colon\R^{+n}\to\R^+$. 
\end{defn}

This definition requires that portfolio generating functions be strictly positive, which could be an impediment to option-pricing applications, since option values can drop to 0 at expiration. However, in Section~\ref{BS} below we show that in practice this apparent difficulty can be overcome. Furthermore, the additive portfolio generation of  \citet{Karatzas/Ruf:2017} might enable direct treatment of non-positive portfolio-generating functions, without the workaround that we use in  Section~\ref{BS}.

The following theorem lays the foundation for the rest of this paper.

\begin{thm}\label{T1} Let $\V\colon\R^{+n}\times[0,T]\to\R^+$ be a portfolio generating function, let  $X_1,\ldots,X_n$ be  price processes, and let $X_0$ be a riskless asset. Then $\V$ generates the portfolio $\p$ for $X_0,X_1,\ldots,X_n$ with weights
\begin{equation}\label{3.5}
\p_i(t) = \frac{X_i(t)D_i\V(X(t),t)}{\V(X(t),t)},
\end{equation}
for $i=1,\ldots,n$ and $t\in[0,T)$, and
\begin{equation}\label{3.5a}
\p_0(t)=1-\sumi\p_i(t),
\end{equation}
for $t\in[0,T)$, and with the drift process given by
\begin{equation}\begin{split}\label{3.6}
d\Phi_\V(t)= -\half&\sumij \frac{D_{ij}\V(X(t),t)}{\V(X(t),t)}X_i(t)X_j(t)\s_{ij}(t)dt\\&-\frac{D_t\V(X(t),t)}{\V(X(t),t)}dt+\g_0(t)\bigg(1-
\sumi\frac{X_i(t)D_i\V(X(t),t)}{\V(X(t),t)}\bigg)dt,
\end{split}\end{equation}
for $t\in[0,T)$, where $\g_0$ is the interest rate process for $X_0$.
\end{thm}

\begin{proof} By \ito's rule,
\begin{align}
d\log \V(X(t),t) &= \sumi D_i\log\V(X(t),t) dX_i(t)+D_t\log\V(X(t),t)dt +\half\sumij D_{ij}\log\V(X(t),t)d\bbrac{X_i,X_j}_t\notag\\
 \begin{split}&=\sumi X_i(t)D_i\log\V(X(t),t) d\log X_i(t)+\half\sumi  X_i(t)D_i\log\V(X(t),t)\s_{ii}(t)dt  \\
&\qquad+D_t\log\V(X(t),t)dt+\half\sumij X_i(t)X_j(t)D_{ij}\log\V(X(t),t)\s_{ij}(t)dt,
\end{split}\label{3.8}
\end{align}
for $t\in[0,T)$, where we applied \ito's rule again for $d\log X_i(t)$ in~\eqref{3.8}. In $D_i\log\V=D_i(\log \circ\V)$, etc., the partial derivatives are applied to the composition of $\log$ and $\V$.

The $\p_i$ defined by~\eqref{3.5} do not necessarily add up to 1, so to define a portfolio we need to use $\p_0$ from~\eqref{3.5a}, and with this portfolio $\p=(\p_0,\p_1,\ldots,\p_n)$ equation~\eqref{3.8} becomes
\begin{align*}
d\log \V(X(t),t) &=\sumi \p_i(t) d\log X_i(t)+\half\sumi \p_i(t)\s_{ii}(t)dt+\frac{D_t\V(X(t),t)}{\V(X(t),t)}dt \\
&\qquad+\half\sumij \frac{D_{ij}\V(X(t),t)}{\V(X(t),t)}X_i(t)X_j(t)\s_{ij}(t)dt-\half\sumij\p_i(t)\p_j(t)\s_{ij}(t)dt\\
&=\sumi \p_i(t) d\log X_i(t)+\g^*_\p(t)dt
+\frac{D_t\V(X(t),t)}{\V(X(t),t)}dt \\
&\qquad+\half\sumij \frac{D_{ij}\V(X(t),t)}{\V(X(t),t)}X_i(t)X_j(t)\s_{ij}(t)dt\\
&=d\log Z_\p(t)-d\Phi_\V(t),
\end{align*}
for $t\in[0,T)$, as in~\eqref{1.2}, \eqref{1.3}, and~\eqref{3.6}, and this satisfies~\eqref{3.4}.
\end{proof}

This theorem shows that a portfolio generating function will generate a portfolio as long as a riskless asset $X_0$ is also included in the portfolio. In the next section we consider a restricted class of portfolio generating functions that do not require the addition of a riskless asset.

\section{Contingent claim functions} \label{contingent}

\citet{Merton:1973} argued that from an economic perspective, homogeneity of degree~1 is a natural property of option-pricing functions  (see \citet{Merton:1973}, Theorem~6). This argument states that a change in the numeraire currency should change the value of an option by the same exchange rate as it changes the values of the assets in the replicating portfolio, so the option value function must be homogeneous of degree~1. \citet{HN:2001,HN2:2001} assert that scale invariance is a fundamental property of any well-posed market model, so that the prices of all tradable assets, including options, will scale equivalently. They show that this would require that option-pricing functions be homogeneous of degree 1, and they use this property to show that the trading strategies used for hedging contingent claims are self-financing (see \citet{HN:2001}, Section~2.3, equations~(2.15) and (2.16)). 

Since self-financing of portfolios is automatic with the proportional definition of portfolio value in~\eqref{1.1.1}, and since the previous section showed that portfolio generating functions need not be homogeneous, we find that there is no mathematical need for homogeneity. Nevertheless, the prevalence of homogeneity in option-pricing theory provides ample rationale for us to investigate its theoretical ramifications. We first need a couple of definitions.

\begin{defn}\label{DH1}
For integers $k\ge0$ and $n\ge1$, a function $V\colon\R^{+n}\to\R^+$ is {\em homogeneous of degree~$k$} if for $c\in\R^+$ and $x\in \R^{+n}$, $V(cx)=c^kV(x)$. A function $V\colon\R^{+n}\times[0,T]\to\R^+$ is {\em homogeneous of degree~$k$} if for $c\in\R^+$, $x\in \R^{+n}$, and $t\in[0,T]$, $V(cx,t)=c^kV(x,t)$.
\end{defn}

\begin{defn}\label{D5.0} 
A {\em contingent claim function}  $\V\colon\R^{+n}\times [0,T]\to\R^+$  is a portfolio generating function that is homogeneous of degree~1. 
\end{defn}

Like Definition \ref{PG1}, this definition restricts contingent claim functions to those with values in $\R^+$, and this excludes  common put and call pricing functions since their terminal values can be 0. However, this apparent exclusion can be accommodated by modifying the terminal value function, as we show below in Section~\ref{BS}.

The following lemma, sometimes referred to as Euler's Theorem for Homogeneous Functions, allows us to explore some of the implications of homogeneity. 

\begin{lem}\label{L1} For an integer $k\ge0$, a $C^1$ function $V\colon\R^{+n}\to\R^+$  is homogeneous of degree~$k$ if and only if 
\begin{equation}\label{2.4}
kV(x)=\sumi x_i D_i V(x),
\end{equation}
for $x\in\R^{+n}$.
\end{lem}

\begin{prop}\label{C0} Let $\V\colon\R^{+n}\times[0,T]\to\R^+$  be a contingent claim function  and let $X_1,\ldots,X_n$ be  price processes. Then $\V$  generates the portfolio  $\p$  for $X_1,\ldots,X_n$ with weights
\begin{equation}\label{3.15}
\p_i(t) = \frac{X_i(t)D_i\V(X(t),t)}{\V(X(t),t)},
\end{equation}
for $i=1,\ldots,n$ and  $t\in[0,T)$, and with  drift process given by
\begin{equation}\label{3.16}
d\Phi_\V(t)= -\half\sumij \frac{D_{ij}\V(X(t),t)}{\V(X(t),t)}X_i(t)X_j(t)\s_{ij}(t)dt-\frac{D_t\V(X(t),t)}{\V(X(t),t)}dt,
\end{equation}
for $t\in[0,T)$.
\end{prop}
\begin{proof}
By Lemma~\ref{L1} with $k=1$, the $\p_i(t)$ in~\eqref{3.15} sum to~1, so the last term in~\eqref{3.6} of Theorem~\ref{T1} vanishes.
\end{proof}

Note the contrast between this proposition, in which a contingent claim function generates a portfolio for $X_1, \ldots, X_n$, with Theorem~\ref{T1}, in which a portfolio generating function generates a portfolio for $X_0, X_1, \ldots, X_n$. We also have a partial converse to Proposition \ref{C0}.
 
\begin{prop}\label{C00} Let  $\V\colon\R^{+n}\times[0,T]\to\R^+$ be a portfolio generating function. If for all price processes $X_1,\ldots,X_n$ the portfolio $\p$ that $\V$ generates  satisfies $\p_0\equiv0$, then $\V$ is a contingent claim function. 
\end{prop}
\begin{proof} Since $X(t)=(X_1(t),\ldots,X_n(t))$ spans $\R^{+n}$ for $t\in[0,T]$, if $\p_0\equiv0$ in~\eqref{3.5} and~\eqref{3.5a} then for $(x,t)\in\R^{+n}\times[0,T)$,
\[
\sumi\frac{x_iD_i\V(x,t)}{\V(x,t)}=\sumi\frac{X_i(t)D_i\V(X(t),t)}{\V(X(t),t)}=\sumi \p_i(t)=1,
\]
 so $\V$ is homogeneous of degree~1  by Lemma~\ref{L1}.
\end{proof}

Propositions~\ref{C0} and~\ref{C00} show that homogeneity of degree~1 is precisely the condition required for a portfolio generating function in order to obviate  the addition of a riskless asset $X_0$ to the portfolio it generates. \citet{Merton:1973} and \citet{HN:2001,HN2:2001} provide economic rationale for homogeneity of degree~1 in option-pricing functions, but here homogeneity of degree~1 is the mathematical condition that ensures that the $\g_0$ term vanishes from~\eqref{3.6}, and hence that the riskless asset $X_0$ vanishes from the portfolio $\p$.

The original theory of  portfolio generating functions defined them as functions of the market weights $\m_1, \ldots, \m_n \in \D^n$,  as in~\eqref{3.10}. Many examples of such generating functions exist, and it is useful to extend these functions to portfolio generating functions in the current setting of Definition~\ref{PG1}. As it happens, these original portfolio generating functions defined on $\m_1, \ldots, \m_n$ can be extended to contingent claim functions defined on $X_1, \ldots, X_n$. 

\begin{lem}\label{L2} Let $V\colon\R^{+n}\to\R^+$ be $C^1$ and homogeneous of degree~1, and let $x\in\R^{+n}$,  $z=x_1+\cdots+x_n$, and $\m=z^{-1}x$. Then, for $i=1,\ldots,n$,
\begin{equation*}
D_iV(x)=V(\m)+D_iV(\m)-\sumj \m_jD_jV(\m).
\end{equation*}
\end{lem}
\begin{proof}   $V$ is homogeneous of degree~1, so for $x\in\R^{+n}$,  $z=x_1+\cdots+x_n$, and $\m=z^{-1}x$,
\begin{equation}\label{3.111a}
V(x)=V(z\m)=zV(\m).
\end{equation}
 Then the calculations 
\[
\dd{z}{x_i}=1,\quad\text{and}\quad \dd{\m_j}{x_i}=\dd{}{x_i}(x_jz^{-1})=z^{-1}(\d_{ij}-\m_j),\vspace{5pt}
\]
for $i,j=1,\ldots,n$, give us
\[
\dd{}{x_i}V(\m)=\sumj D_jV(\m)\dd{\m_j}{x_i}=z^{-1}\Big(D_iV(\m)-\sumj \m_jD_jV(\m)\Big),
\]
for $i=1,\ldots,n$. If we differentiate both sides of~\eqref{3.111a}, we have 
\begin{align*}
D_iV(x)&=\dd{}{x_i}V(x)=\dd{}{x_i}\big(zV(\m)\big)\notag\\
&= V(\m)+z\dd{}{x_i}V(\m)\notag\\
&=V(\m)+D_iV(\m)-\sumj \m_jD_jV(\m),
\end{align*}
for $i=1,\ldots,n$. 
\end{proof}

\begin{prop}\label{P3} Suppose that for an open neighborhood $U$ of $\D^n\subset\R^n$ the function $\S\colon U\times[0,T]\to\R^+$ is $C^{2,1}$ such that $x_iD_i\log\S(x,t)$ is bounded for $i=1,\ldots,n$ and $(x,t)\in\D^n\times[0,T]$. Let $\SS\colon\R^{+n}\times[0,T]\to\R^+$ be defined by
\begin{equation}\label{3.111}
\SS(x,t)=z\S(z^{-1}x,t)=z\S(\m,t),
\end{equation}
for $x\in\R^{+n}$, $t\in[0,T]$, $z=x_1+\cdots+x_n$, and $\m=z^{-1}x\in\D^n$. Then $\SS$ is a contingent claim function, $\S$ and $\SS$ generate the same portfolio, and for price processes $X_1,\ldots,X_n$ the drift processes $\Phi_\SS$ of~(\ref{3.4}) and $\Theta_\S$ of~(\ref{3.10}) are equal.
\end{prop}

\begin{proof}  Since $\S$ is $C^{2,1}$, it follows from~\eqref{3.111} that $\SS$ is also $C^{2,1}$. For $c\in\R^+$,
\begin{equation}\label{3.112}
\SS(cx,t)=cz\S((cz)^{-1}cx,t)=cz\S(z^{-1}x,t)=c\SS(x,t),
\end{equation}
for $x\in\R^{+n}$ and $t\in[0,T]$, so $\SS$ is homogeneous of degree~1. For $(\m,t)\in\D^n\times[0,T]$ we have  $\SS(\m,t)=\S(\m,t)$, so by Lemma~\ref{L2},
\[
D_i\SS(x,t)=\S(\m,t)+D_i\S(\m,t)-\sumj\m_jD_j\S(\m,t),
\]
for $i=1,\ldots,n$, and $t\in[0,T]$. By~\eqref{3.5}, the weight processes  for $\SS$ are defined by
\begin{align*}
\frac{x_iD_i\SS(x,t)}{\SS(x,t)}&=\frac{x_iD_i\SS(x,t)}{z\S(\m,t)}\\
&=\m_i\Big(1+ D_i\log\S(\m,t)-\sumj\m_jD_j\log\S(\m,t)\Big),
\end{align*}
for $i=1,\ldots,n$, and $t\in[0,T]$, which are  equal to the weight processes for $\S$ in~\eqref{3.9}, so $\SS$ is a contingent claim function, and $\S$ and $\SS$ generate the same portfolio $\p$. Hence, for price processes $X_1,\ldots,X_n$, the term  $d\log Z_\p(t)$ in~\eqref{3.10} will be  equal to the corresponding term in~\eqref{3.4}, so
\begin{align*}
d\Phi_\SS(t)&=d\log Z_\p(t) -d\log\SS(X(t),t)\\
&=d\log Z_\p(t)-d\log\big(Z(t)\S(\m(t),t)\big)\\
&=d\log\big(Z_\p(t)/Z(t)\big)-d\log\S(\m(t),t)\\
&=d\Theta_\S(t),
\end{align*}
for $t\in[0,T]$.
\end{proof}

The Gibbs-Shannon entropy function $\S\colon\D^n\to \R^+$, defined by
\begin{equation}\label{3.12a}
\S(x)=-\sumi x_i\log x_i,
\end{equation}
for  $x\in\D^n$, served as a prototype for portfolio generating functions in \citet{F:div}, Definition~4.2. For a market $X_1,\ldots,X_n$, the entropy function $\S$  generates the portfolio $\p$ with weight processes 
\begin{equation}\label{3.12}
\p_i(t)=-\frac{\m_i(t)\log \m_i(t)}{\S(\m(t))},
\end{equation}
for $i=1,\ldots,n$ and $t\in[0,T]$, where the $\m_i$ are the market weight processes (see also \citet{F:pgf}, Example~4.2, and \citet{F:2002}, Section~2.3). In the following example we apply \eqref{3.111} to calculate the extension of $\S$ to a contingent claim function $\SS$, and we can see that these two functions generate the same portfolio for price processes $X_1,\ldots,X_n$.

\begin{exa}\label{Ex-ent} {\em The  extended entropy function.} For $\S$ as in~\eqref{3.12a}, let us write $\S(x,t)=\S(x)$, for $x\in\D^n$ and $t\in[0,T]$. Then, following Proposition~\ref{P3}, we extend $\S$ to a contingent claim function $\SS$ as in~\eqref{3.111}. For $x\in\R^{+n}$, $z=x_1+\cdots+x_n$, and $\m=z^{-1}x$, we have $\m\in\D^n$ and 
\[
\SS(x,t)=z\S(\m,t)=-z\sumi\m_i\log\m_i=z\log z-\sumi x_i\log x_i,
\]
for $t\in[0,T]$. Then, for price processes $X_1,\ldots,X_n$, the {\em extended entropy function} $\SS$ generates the portfolio $\p$  with weights
\[
\p_i(t)=\frac{X_i(t)D_i \SS(X(t),t)}{\SS(X(t),t)}=\frac{X_i(t)\big(\log Z(t)-\log X_i(t)\big)}{\SS(X(t),t)}=-\frac{\m_i(t)\log\m_i(t)}{\S(\m(t),t)},
\]
for $i=1,\ldots,n$, and $t\in[0,T]$, and this agrees with~\eqref{3.12}, as required by Proposition~\ref{P3}. \qed
\end{exa}

\section{Replicability} \label{cc}

The concept of {\em replicability} was introduced by \citet{Merton:1973}, where he considered combinations of an asset, an option on the asset, and a riskless bond (see the paragraph that starts at the bottom of page~168 of \citet{Merton:1973}). Although it was not mentioned by name, replicability was central to the option-pricing model of \citet{Black/Scholes}, where they showed that the derivative of their pricing function determined the amount of an asset needed to hedge --- or replicate --- the option (see equation~(14) of \citet{Black/Scholes}). Here we extend these concepts to portfolio generating functions and contingent claim functions, and then in Section~\ref{BS} below we return to the Black-Scholes model to apply these results.

Replicability was originally considered in the context of arbitrage-free markets, however, arbitrage might exist in our setting  (see \citet{F:2002}, Section~3.3, \citet{FKR:2018}, Section~5, or \citet{KK:2021}, Section~2.2.2). Accordingly, here we consider replicability {\em per se,} with no connection to arbitrage. The idea of replicability is that the value of a replicable portfolio generating function $\V$ will at all times be the same as the value $Z_\p$ of the portfolio $\p$ it generates.  The replicability of a portfolio generating function depends not only on the function itself, but also on the price processes that comprise the portfolio it generates, and we must define it accordingly.
   
\begin{defn}\label{D5b} Let $X_1,\ldots,X_n$ be price processes, let $X_0$ be a riskless asset, let $\V\colon\R^{+n}\times[0,T]\to\R^+$ be a portfolio generating function, let $\p$ be the portfolio $\V$ generates, and let $Z_\p$ be the corresponding value process. Then $\V$ is  {\em replicable} for $X_0,X_1,\ldots,X_n$ if
\begin{equation}\label{4.0}
d\log\V(X(t),t)=d\log Z_\p(t),
\end{equation}
for  $t\in[0,T)$. Similarly, a contingent claim function $\V$ is {\em replicable} for $X_1,\ldots,X_n$ if it satisfies \eqref{4.0}. 
\end{defn} 
 
\begin{prop}\label{prep}  Let $X_1,\ldots,X_n$ be price processes, let $X_0$ be a riskless asset, let $\V\colon\R^{+n}\times[0,T]\to\R^+$ be a portfolio generating function, and let $\Phi_\V$ be the corresponding drift process. Then $\V$ is replicable for $X_0,X_1,\ldots,X_n$ if and only if
\begin{equation}\label{rep}
\Phi_\V(t)=0,
\end{equation}
for  $t\in[0,T)$. 
 \end{prop}
\begin{proof} This follows directly from Definitions~\ref{D3} and \ref{D5b}.
\end{proof}

Note that \eqref{rep} also applies to contingent claim functions since they are portfolio generating functions. 
  
\begin{cor}\label{P2a}  Let  $\V\colon\R^{+n}\times[0,T]\to\R^+$ be a portfolio generating function, let $X_1,\ldots,X_n$ be price processes, and let $X_0$ be a riskless asset. Then $\V$ is replicable for $X_0,X_1,\ldots,X_n$ if and only if
 \begin{equation}\begin{split}\label{repgen}
\half\sumij \s_{ij}(t)X_i(t)X_j(t)&D_{ij}\V(X(t),t)+D_t\V(X(t),t)\\&+\g_0(t)\Big(\sumi X_i(t)D_i\V(X(t),t)-\V(X(t),t)\Big)=0,
\end{split}\end{equation}
for  $t\in[0,T)$. 
 \end{cor}
\begin{proof} This follows from Proposition~\ref{prep} and~\eqref{3.6} of Theorem~\ref{T1}.
 \end{proof}
 
Equation~\eqref{repgen} becomes exactly the (generalized) {\em Black-Scholes equation} for the price of a European option when $\g_0=r$, the riskless interest rate (see \citet{Black/Scholes}, equation~(7)). The original Black-Scholes equation, which is given in~\eqref{5.1} below, had constant coefficients, whereas in~\eqref{repgen} the coefficients are stochastic. Similar multivariate Black-Scholes equations have appeared  in, e.g.,  \citet{Duffie:2001}, \citet{Carmona2006}, and \citet{Guillaume2019}.

Corollary~\ref{P2a} can be simplified for contingent claim functions, since there is no need for the riskless asset~$X_0$.

\begin{cor}\label{P1} Let   $\V\colon\R^{+n}\times [0,T]\to\R^+$ be a contingent claim function and let $X_1,\ldots,X_n$ be price processes. Then $\V$ is replicable for $X_1,\ldots,X_n$ if and only if
\begin{equation}\label{4.3}
 \half\sumij \s_{ij}(t)X_i(t)X_j(t)D_{ij}\V(X(t),t)+D_t\V(X(t),t)=0,
 \end{equation}
 for $t\in[0,T)$.
 \end{cor}
\begin{proof}  Lemma~\ref{L1} implies that for a contingent claim function the $\g_0$ term in \eqref{repgen} vanishes.
 \end{proof}
 
Equation~\eqref{4.3}  is equivalent to~(2.15) in \citet{HN:2001}, where homogeneity of degree~1 is required of all option-pricing functions. Equation~(2.15) of \citet{HN:2001} is used to show that the replicating portfolio for a contingent claim is self-financing, whereas in our setting all portfolios defined by~\eqref{1.1.1} are self-financing.
 
As we mentioned above, replicability depends on the price processes as well as the generating function. The following examples present some replicable and non-replicable portfolio generating functions, and they show that replicability depends as much on the nature of the processes $X_1,\ldots,X_n$ as on the form of the function $\V$. The first example presents a contingent claim function that is a generalization of the geometric mean of the price processes $X_1, \ldots, X_n$. This function is replicable only for price processes with constant covariance processes. 

\begin{exa}\label{X1} {\em A conditionally replicable contingent claim function.} 
For real constants $p_1,\ldots,p_n$, with $p_1+\cdots+p_n=1$, let $\V\colon\R^{+n}\times[0,T]\to\R^+$ be the $C^{2,1}$ function defined by 
\[
\V(x,t)=x_1^{p_1}\cdots x_n^{p_n},
\]
for $(x,t) \in \R^{+n} \times [0,T]$. This function is based on the portfolio generating function $\S$ in Example~3.1.6(4), in \citet{F:2002}. We see that $\V$ is homogeneous of degree~1, so it is a contingent claim function and for  price processes $X_1,\ldots,X_n$ it generates the portfolio $\p$ with weights
\[
\p_i(t)=\frac{X_i(t)D_i\V(X(t),t)}{\V(X(t),t)}=\frac{X_i(t)\big(p_iX_1^{p_1}(t)\cdots X_i^{p_i-1}(t)\cdots X_n^{p_n}(t)\big)}{X_1^{p_1}(t)\cdots X_n^{p_n}(t)}=p_i,
\]
for $i=1,\ldots,n$ and $t\in[0,T]$. Here $\p$ is a {\em constant-weighted} portfolio,  and by  Theorem~\ref{T1} we have the decomposition~\eqref{3.4} with drift process given by
\begin{align}
d\Phi_\V(t)&=\half\bigg(\sumi p_i\s_{ii}(t)-\sumij p_i p_j\s_{ij}(t)\bigg)dt\label{4.4.0}\\
&=\g_\p^*(t)dt,\label{4.4.1}
\end{align}
for $t\in[0,T]$. If, for example, $p_i>0$, for $i=1,\ldots,n$, and the covariance matrix $\big(\s_{ij}(t)\big)_{\{i,j=1,\ldots,n\}}$ is nonsingular, then $\g^*_\p(t)>0$ (see \citet{F:2002}, Proposition~1.3.7), so $d\Phi_\V(t)>0$,  and, by Corollary~\ref{P1}, $\V$ is not replicable. 

It would be convenient to define $\Vt\colon\R^{+n}\times[0,T]\to\R^+$ by
\begin{equation} \label{4.4}\begin{split}
\Vt(x,t) &= \V(x,t)e^{\Phi_\V(t)}\\
 & = \V(x,t)\exp\Big(\intt\g_\p^*(s)ds\Big),
\end{split}\end{equation}
for $(x,t)\in\R^{+n}\times[0,T]$, which would compensate for the term~\eqref{4.4.1}. However, $\Vt$ must be a function, and in general $\Phi_\V$ might be a random process rather than a function, in which case $\Vt$ will not be well defined. However, if the covariance processes $\s_{ij}$ in~\eqref{4.4.0} are constant, then $\g^*_\p(t)=\g^*_\p(0)$ for $t\in[0,T]$, so
\begin{equation}\label{4.4.2}
\Vt(x,t)=\V(x,t)e^{\g_\p^*(0) t},
\end{equation}
for $(x,t)\in\R^{+n}\times[0,T]$, and $\Vt$ is a $C^{2,1}$ function that generates the same portfolio as $\V$. Hence, for $t\in[0,T]$, 
\begin{align*}
d\log Z_\p(t)&=d\log\V(X(t),t)+d\Phi_\V(t)\\
&=d\log\V(X(t),t)+\g^*_\p(0)dt\\
&=d\log\Vt(X(t),t),
\end{align*}
by~\eqref{4.4.2}, so  $\Vt$ is replicable for price processes with constant covariances. We should note here that although we changed the contingent claim function from $\V$ to $\Vt$ in~\eqref{4.4.2}, the portfolio weights $\p_i\equiv p_i$, for $i=1,\ldots,n$, are the same for both functions. \qed
\end{exa}

\begin{rem} 
Although~\eqref{4.4} does not always produce a replicable contingent claim function, if we could obtain a reasonably constant approximation of $\g^*_\p(t)$, then we could use that to estimate $\Vt(x,t)$ according to~\eqref{4.4}. This could be done, for example, by assuming that $\g^*_\p(t) = \g^*_\p(0)$, or by estimating $\g^*_\p(t)$ directly using~\eqref{1.3}. More generally, an estimate of a non-zero drift process $\Phi_\V$ together with the construction $\Vt = \V e^{\Phi_\V}$ might be useful in the practice of option pricing, where parameters are often not precise, but must be estimated using past data and models for the volatility of the underlying assets.
\qed
\end{rem}

The following two examples examine a contingent claim function similar to the $L^p$-norm and describe how this function is not replicable in its general form, but can be made replicable in a special case with $\s_{ij} = 0$ for $i \neq j$.

\begin{exa}\label{X2}  {\em A non-replicable contingent claim function.}  For $p\in(0,1)$ let
\begin{equation}\label{4.4.3}
\V(x,t)=\big(x_1^p+\cdots+x_n^p\big)^{1/p},
\end{equation}
for $(x,t) \in \R^{+n} \times [0,T]$. This function is based on the measure of diversity $\DD_p$ in Example~3.4.4 of \citet{F:2002}. We see that $\V$ is $C^{2,1}$ and homogeneous of degree~1, so it is a contingent claim function and for price processes $X_1,\ldots,X_n$ it generates the portfolio $\p$ with weights
\[
\p_i(t)=\frac{X_i(t)D_i\V(X(t),t)}{\V(X(t),t)}=\frac{X_i^p(t)}{X_1^{p}(t)+\cdots +X_n^{p}(t)},
\]
for $i=1,\ldots,n$  and $t\in[0,T]$. In this case the decomposition~\eqref{3.4} has drift process given by
\[
d\Phi_\V(t)=(1-p)\g_\p^*(t)dt,
\]
for $t\in[0,T]$ (see \citet{F:2002}, Example~3.4.4), so as in Example~\ref{X1} $\V$ is not replicable. By modifying $\V$ as in~\eqref{4.4} we have 
\[
\Vt(x,t)=\V(x,t)\exp\Big((1-p)\intt\g_\p^*(s)ds\Big),
\]
for $(x,t)\in\R^{+n}\times[0,T]$, but since the portfolio weights $\p_i$ are variable, $\g^*_\p$ will not be constant even if all the covariances are. Because $\g^*_\p$ is a process rather than a function, this modification $\Vt$ fails to create a replicable contingent claim function. \qed
\end{exa}

Although in Example~\ref{X2} we cannot use the simple adjustment of Example~\ref{X1}, nevertheless with constant covariances the contingent claim function $\V$ of~\eqref{4.4.3} can be modified to make it replicable. This modification can be constructed by using a standard method for solving Black-Scholes type equations. The method consists in  using a logarithmic transformation to convert~\eqref{4.3} into a heat equation, and then finding a convolutional solution to this heat equation (see \citet{Evans:2010}, Section~2.3, Theorem~1). In the following example we carry out this modification of $\V$ in~\eqref{4.4.3}  for the simple case of $p=1/2$ and $\s_{ij}=0$ for $i\ne j$.

\begin{exa}\label{X4a}  {\em A replicable version of Example~\ref{X2}.}  Let $p=1/2$ and let
\[
d\log X_i(t)=\g_i(t)dt+ \s_i dW_i(t),
\]
for  $i=1,\ldots,n$ and $t\in[0,T]$, as in~\eqref{1.0} with constant variance rates $\s^2_i>0$. In this case~\eqref{4.3} will be
\begin{equation}\label{4.9}
\dd{}{t}\V(x,t)+\half\sumi \s_i^2 x^2_i \dd{^2}{x_i^2}\V(x,t)=0,
\end{equation}
for $(x,t)\in\R^{+n}\times[0,T]$,  and $\V \colon \R^{+n} \times [0,T] \to \R^+$ is $C^{2,1}$. We wish to solve this equation with the terminal condition
\begin{equation}\label{4.9a}
\V(x,T)=\big(x_1^{1/2}+\cdots+x_n^{1/2}\big)^2,
\end{equation}
for $x\in\R^{+n}$, which is the same function as in~\eqref{4.4.3} for $p=1/2$.  

The solution to~\eqref{4.9} with terminal condition~\eqref{4.9a} is
\begin{equation}\label{4.9b}
\V(x,t)= \sumi x_i+\sum^n_{\substack{i,j=1\\ i\ne j}}e^{\s_i^2(t-T)/8} e^{\s_j^2(t-T)/8}x_i^{1/2}x_j^{1/2},
\end{equation}
for $(x,t)\in\R^{+n}\times[0,T]$, and the derivation of this solution appears in Appendix~\ref{appendix}. $\V$ generates the portfolio $\p$ with weights
\[
\p_i(t)=\frac{X_i(t)+ e^{\s_i^2(t-T)/8}X^{1/2}_i(t)\sum^n_{\substack{j=1\\ \hspace{-25pt}j\ne i}} e^{\s_j^2(t-T)/8}X_j^{1/2}(t)}{\V(X(t),t)},
\]
for $t\in[0,T]$ and $i=1,\ldots,n$. Since~\eqref{4.9} is satisfied, the portfolio $\p$ satisfies~\eqref{4.3}, and $\V$ is replicable for $X_1,\ldots,X_n$. This example is like a ``backward'' problem in that the terminal value function is the same as in Example~\ref{X2}, i.e., with $p=1/2$ in~\eqref{4.4.3}, but the weights are different. This contrasts with Example~\ref{X1}, which is like a ``forward'' problem in that the weights remain the same but the contingent claim function changes.
 \qed
\end{exa}

\section{Portfolio generating functions and contingent claim functions} \label{GFCCF}

It is common to require that option-pricing functions be homogeneous of degree 1, so that these functions exist within the class of contingent claim functions \citep{Merton:1973,HN:2001}. In this section we show how to transform an inhomogeneous portfolio generating function $\V\colon\R^{+n}\times[0,T]\to\R^+$ into a contingent claim function $\Vh\colon\R^+\times\R^{+n}\times[0,T]\to\R^+$, which is homogeneous of degree~1 by definition. This method will be applied in the next section to find contingent claim functions for option-pricing problems with inhomogeneous terminal value functions. The method consists in first dividing each variable $x_i$ by a new variable $x_0 \in \R^+$, which is equivalent to discounting the asset prices $X_i$ by a riskless asset $X_0$. This converts the original function into a function that is homogeneous of degree 0 for $(x_0, x)\in\R^+\times\R^{+n}$, and then this new function is multiplied by $x_0$, which raises the function to be homogeneous of degree 1. We now show that a portfolio generating function transformed in this manner satisfies the criteria for a contingent claim function. 

\begin{prop}\label{T2} Let $\V\colon\R^{+n}\times[0,T]\to\R^+$ be a portfolio generating function and let  $\Vh\colon\R^+\times\R^{+n}\times[0,T]\to\R^+$ be defined by
\begin{equation}\label{8.0.0} 
\Vh(x_0,x,t)=x_0\V(x_0^{-1}x,t),
\end{equation}
for $(x_0,x,t)\in\R^+\times\R^{+n}\times[0,T]$. Then $\Vh$ is a contingent claim function.
\end{prop}

\begin{proof} It follows from \eqref{8.0.0} that $\Vh$ is $C^{2,1}$ on $\R^+\times\R^{+n}\times[0,T)$, and for $a\in\R^+$,
\[
\Vh(ax_0,ax,t)=ax_0\V((ax_0)^{-1}ax,t)=ax_0\V(x_0^{-1}x,t)=a\Vh(x_0,x,t),
\]
for $(x_0,x,t)\in\R^+\times\R^{+n}\times[0,T]$, so $\Vh$ is homogeneous of degree~1. To show that $\Vh$ is a contingent claim function, we must show that its weight processes are bounded. From~\eqref{8.0.0}  we see that for $(x_0,x,t)\in\R^+\times\R^{+n}\times[0,T)$,
\begin{align}
D_i\Vh(x_0,x,t)&=\dd{}{x_i}\Vh(x_0,x,t)=\dd{}{x_i}\big(x_0\V(x_0^{-1}x,t)\big)\notag\\
&=x_0\dd{x_0^{-1}x_i}{x_i}D_i\V(x_0^{-1}x,t)\label{8.0.1}\\
&=D_i\V(x_0^{-1}x,t),\notag
\end{align}
for $i=1,\ldots,n$, where~\eqref{8.0.1} results from an application of the chain rule as in~\eqref{1.6} and the fact that $D_i\V$ is the partial derivative of $\V$ with respect to its $i$th argument. Hence, for $(x_0,x,t)\in\R^+\times\R^{+n}\times[0,T)$, we have
\begin{equation}\label{8.1.0}
\frac{x_i D_i\Vh(x_0,x,t)}{\Vh(x_0,x,t)}=\frac{x_iD_i\V(x_0^{-1}x,t)}{x_0\V(x_0^{-1}x,t)}=\frac{x_0^{-1}x_iD_i\V(x_0^{-1}x,t)}{\V(x_0^{-1}x,t)},
\end{equation}
for  $i=1,\ldots,n$. Since terms of the form $x_0^{-1}x\in\R^{+n}$ span $\R^{+n}$, we see that $x_iD_i\Vh(x_0,x,t)/\Vh(x_0,x,t)$ is bounded on $\R^+\times\R^{+n}\times[0,T)$ if and only if $x_iD_i\V(x,t)/\V(x,t)$ is bounded on $\R^{+n}\times[0,T)$, which it is since $\V$ is a portfolio generating function. Since $\Vh$ is homogeneous of degree~1, for $(x_0,x,t)\in\R^+\times\R^{+n}\times[0,T]$, 
\[
\sum_{i=0}^n \frac{x_i D_i\Vh(x_0,x,t)}{\Vh(x_0,x,t)}=1,
\]
by Lemma~\ref{L1}, and this implies that $x_0D_0\Vh(x_0,x,t)/\Vh(x_0,x,t)$ is also bounded. Hence, $\Vh$ is a contingent claim function.
\end{proof}

This proposition provides a method of converting a portfolio generating function $\V\colon\R^{+n}\times[0,T]\to\R^+$ into a contingent claim function  $\Vh\colon\R^+\times\R^{+n}\times[0,T]\to\R^+$. Although the domains of definition for these two functions differ, both functions generate portfolios for the same price processes $X_1,\ldots,X_n$ and  riskless asset $X_0$. Proposition~\ref{T2} depends only on mathematical arguments, not economic arguments; it first induces homogeneity of degree~0, and then reverses the process to raise the resulting function to homogeneity of degree~1. The mathematical nature of Proposition~\ref{T2}   is in contrast to \citet{HN:2001}, where economic arguments are presented to convert an ostensibly inhomogeneous function  into a modified function  that is homogeneous of degree~1 (see equations~(2.4) and~(2.5) in Section~2.1  of \citet{HN:2001}). The versatility of Proposition~\ref{T2} is illustrated in Example~\ref{XA3} below.

We would like to show that the conversion~\eqref{8.0.0} preserves replicability, and Lemmata~\ref{L5} and~\ref{L6} leading to Proposition~\ref{P10} below accomplish this. In the next section, these results will allow us to find contingent claim functions that solve the option-pricing problem even for inhomogeneous terminal value functions. We first consider replicability of $\V$ and $\Vh$ for $1,X_1,\ldots,X_n$, where the riskless asset $X_0\equiv1$.

\begin{lem}\label{L5}  Let $\V\colon\R^{+n}\times[0,T]\to\R^+$ be a portfolio generating function, let $\Vh\colon\R^+\times\R^{+n}\times[0,T]\to\R^+$ be the contingent claim function defined by~\eqref{8.0.0}, and let $X_1,\ldots,X_n$ be price processes. Then $\Vh$  is replicable for $1,X_1,\ldots,X_n$ if and only if $\V$ is replicable for $1,X_1,\ldots,X_n$.
\end{lem}

\begin{proof} 
From~\eqref{8.0.0} we have
\[
 x_0\V(x_0^{-1}x,t)=\Vh(x_0,x,t)= x_0\Vh(1,x_0^{-1}x,t),
\]
for $(x_0, x, t)\in\R^+\times\R^{+n}\times[0,T]$, since $\Vh$ is homogeneous of degree~1. Hence,
\begin{equation*}\label{8.1.2}
\V(x,t)=\Vh(1,x,t),
\end{equation*}
for $(x,t)\in\R^{+n}\times[0,T]$, so the corresponding partial derivatives for $\V$ and $\Vh(1,\,\cdot\,,\,\cdot\,)$ agree. For the riskless asset $X_0\equiv1$, the interest rate process $\g_0\equiv0$, so~\eqref{repgen} becomes
\[
\half\sum_{i,j=1}^n \s_{ij}(t)X_i(t)X_j(t)D_{ij}\V(X(t),t)+D_t\V(X(t),t)=0,
\]
 for $t\in[0,T)$. Since the corresponding partial derivatives  for $\V$ and $\Vh(1,\,\cdot\,,\,\cdot\,)$ agree, this equation will be satisfied if and only if 
\begin{equation*}
 \half\sumij \s_{ij}(t)X_i(t)X_j(t)D_{ij}\Vh(1,X(t),t)+D_t\Vh(1,X(t),t)=0,
 \end{equation*}
  for $t\in[0,T)$. Since $\s_{i0}=\s_{0j}\equiv0$, for $i,j=0,\ldots,n$, we can rewrite this equation as
 \begin{equation*}
 \half\sum_{i,j=0}^n\s_{ij}(t)X_i(t)X_j(t)D_{ij}\Vh(1,X(t),t)+D_t\Vh(1,X(t),t)=0,
\end{equation*}
 for $t\in[0,T)$, which is equivalent to~\eqref{4.3} in this case. Then Corollaries~\ref{P2a} and~\ref{P1} imply that $\Vh$ is replicable for $1,X_1,\ldots,X_n$ if and only if $\V$ is replicable for $1,X_1,\ldots,X_n$.
\end{proof}

Now we show that replicability of a contingent claim function is preserved by multiplication of the price processes $X_1,\ldots,X_n$ by a riskless asset $X_0$.

\begin{lem}\label{L6} Let $\V\colon\R^{+n}\times[0,T]\to\R^+$ be a contingent claim function, let $X_1,\ldots,X_n$ be price processes, and let $X_0$ be a riskless asset. Then $\V$ is replicable for  $X_1,\ldots,X_n$ if and only if it is replicable for  $X_0X_1,\ldots,X_0X_n$.
\end{lem}

\begin{proof} Since $\V$ is a contingent claim function, it is homogeneous of degree~1, so for $c\in\R^+$ and $(x,t)\in\R^{+n}\times[0,T]$,
\begin{equation*}
D_i\V(cx,t)=\dd{}{cx_i}\V(cx,t)=c^{-1}\dd{}{x_i}\big(c\V(x,t)\big)=D_i\V(x,t),
\end{equation*}
for $i=1,\ldots,n$, and hence,
\begin{equation*}
D_{ij}\V(cx,t)=\dd{}{cx_j}D_i\V(cx,t)=c^{-1}\dd{}{x_j}D_i\V(x,t)=c^{-1}D_{ij}\V(x,t),
\end{equation*}
for $i,j=1,\ldots,n$, as well as
\begin{equation*}
D_t\V(cx,t)=D_tc\V(x,t)=cD_t\V(x,t).
\end{equation*}
 If follows that
\begin{align*}
 \half\sumij \s_{ij}(t) & X_0(t)X_i(t)X_0(t)X_j(t)D_{ij}\V(X_0(t)X(t),t)+D_t\V(X_0(t)X(t),t)\\
&= \half\sumij \s_{ij}(t)X_0(t)X_i(t)X_j(t)D_{ij}\V(X(t),t)+X_0(t)D_t\V(X(t),t)\\
&=X_0(t)\bigg(\half\sumij \s_{ij}(t)(t)X_i(t)X_j(t)D_{ij}\V(X(t),t)+D_t\V(X(t),t)\bigg),
 \end{align*}
 for $t\in[0,T]$. Since $X_0>0$, Corollary~\ref{P1} implies that $\V$ is replicable for $X_0X_1,\ldots,X_0X_n$ if and only if it  is replicable for $X_1,\ldots,X_n$.
\end{proof}

Lemmata~\ref{L5} and~\ref{L6} can be combined to prove the following proposition, which we use for option pricing in the next section.

\begin{prop}\label{P10} Let  $\V\colon\R^{+n}\times[0,T]\to\R^+$ be a portfolio generating function, let $X_1,\ldots,X_n$ be price processes, let $X_0$ be a riskless asset, and let $\Vh\colon\R^+\times\R^{+n}\times[0,T]\to\R^+$ be the contingent claim function defined by~\eqref{8.0.0}.   Then $\Vh$ is replicable for $X_0,X_1,\ldots,X_n$ if and only if $\V$ is replicable  for $1,X_0^{-1}X_1,\ldots,X_0^{-1}X_n$.
\end{prop}
\begin{proof} Lemma~\ref{L5} shows that $\V$ is replicable for  $1,X_0^{-1}X_1,\ldots, X_0^{-1}X_n$ if and only if $\Vh$  is replicable for  $1,X_0^{-1}X_1,\ldots, X_0^{-1}X_n$.  Lemma~\ref{L6}, with $X_1,\ldots,X_n$ replaced by $1,X_0^{-1}X_1,\ldots,X_0^{-1}X_n$, shows that $\Vh$ is  replicable for $1,X_0^{-1}X_1,\ldots, X_0^{-1}X_n$ if and only if it is replicable for  $X_0,X_1,\ldots,X_n$. 
\end{proof}

In the next section, we shall see that replicable contingent claim functions are central to option pricing, and this proposition along with~\eqref{8.0.0} will allow us to construct contingent claim functions from inhomogeneous portfolio generating functions in a manner that preserves replicability. 

\section{Option pricing and the Black-Scholes model} \label{BS}

For our purposes, the {\em option-pricing problem} for {\em European} options is: given price processes $X_1,\ldots,X_n$ and a continuous  {\em terminal value function} $f\colon \R^{+n}\to\R^+$, find either ({\em i}) a contingent claim function $\V\colon\R^{+n}\times[0,T]\to\R^+$ that is replicable for $X_1,\ldots,X_n$ with $\V(x,T)=f(x)$, for $x\in\R^{+n}$; or ({\em ii}) a contingent claim function $\Vh\colon\R^+\times\R^{+n}\times[0,T]\to\R^+$ that is replicable for $X_0,X_1,\ldots,X_n$ with $\V(1,x,T)=f(x)$, for $x\in\R^{+n}$, where $X_0$ is a riskless asset with $X_0(T)=1$.
 
Option pricing is based on the economic rationale that at all times the price of an option must equal the value of a portfolio with the same terminal value, or else an opportunity for arbitrage is created. According to this rationale, a replicable contingent claim function is equal to the value of an option with expiration payout equal to the terminal value of that function. Although  arbitrage might exist among our price processes, replicability is central to the concept of option pricing, so we require it here (see \citet{Merton:1973}).

If a terminal value function $f$ is homogeneous of degree~1, then it might be possible to find a contingent claim function that solves the option-pricing problem, as in Examples~\ref{X1} and~\ref{X4a}. However, if $f$ is not homogeneous of degree~1, then a contingent claim function cannot be used directly to solve the option-pricing problem, and we must use ({\em ii}) above. Here we derive a procedure to add a riskless asset $X_0$ to the  price processes $X_1,\ldots,X_n$ to solve the option-pricing problem for an inhomogeneous terminal value function.  

We wish to solve the option-pricing problem for the price processes $X_1,\ldots,X_n$, a riskless asset $X_0$ with $X_0(T)=1$, and the terminal value function $f\colon \R^{+n}\to\R^+$. Our procedure has three steps:

\begin{enumerate}
\item\label{s1} Solve~\eqref{repgen} with $\g_0\equiv0$ to find a portfolio generating function $\V\colon\R^{+n}\times[0,T]\to\R^+$ that is replicable for price processes $1,X_0^{-1}X_1,\ldots,X_0^{-1}X_n$ with terminal value $\V(x,T)=f(x)$, for $x\in\R^{+n}$.

\item\label{s2} Define $\Vh(x_0,x,t)=x_0\V(x_0^{-1}x,t)$, for $(x_0,x,t)\in\R^+\times\R^{+n}\times[0,T]$, as in~\eqref{8.0.0}.

\item\label{s3} Apply Proposition~\ref{P10} to show that $\Vh\colon\R^+\times\R^{+n}\times[0,T]\to\R^+$ is a  contingent claim function which is replicable  for $X_0,X_1,\ldots,X_n$ with terminal value $\Vh(1,x,T)=\V(x,T)=f(x)$, for $x\in\R^{+n}$.
\end{enumerate} 

The only difficult step here is the first, where it is necessary to solve a differential equation of the form~\eqref{repgen} with $\g_0\equiv0$. This equation is of the form of the Feynman-Kac equation (see e.g., \citet{Duffie:2001}, Chapter~5, Section~I, or Appendix~E), and can often be transformed into a heat equation (see, e.g., \citet{Evans:2010}, Section~2.3). In an application, if the price processes $X_0^{-1}X_i$  are of the same generic form as $X_i$, then the processes $1,X_0^{-1}X_1,\ldots,X_0^{-1}X_n$ in Step~1 can be replaced by  $1,X_1,\ldots,X_n$ in the differential equation~\eqref{repgen}, which in any case will have $\g_0\equiv0$. We now apply this procedure to the option-pricing model of \citet{Black/Scholes}.

The Black-Scholes model is a solution to the option-pricing problem for a European call option with a riskless asset $X_0$ of the form~\eqref{1.1a} and a price process $X$ of the form~\eqref{1.0} with 
\begin{equation}\label{11}
d\log X_0(t)=r\,dt\quad\text{ and }\quad d\log X(t)= \g(t)dt+ \s\, dW(t),
\end{equation}
for $t\in[0,T]$, where $r$ and $\s>0$ are constants. \citet{Black/Scholes} showed that  $\V$ follows the differential equation
\begin{equation}\label{5.1}
\dd{\V}{ t}+\half\s^2 x^2 \dd{^2 \V}{x^2}+rx\dd{\V}{x}-r\V=0,
\end{equation}
for $(x,t)\in\R^{+}\times[0,T)$, with terminal value
\begin{equation}\label{5.2}
f(x)= (x-K)^+,
\end{equation}
for $x\in\R^+$, where $K\in\R^+$ is a constant called the {\em strike price} of the call. (Let us note that~\eqref{5.1} is of the form~\eqref{repgen}, where a riskless asset of the form $X_0(t)=e^{r(t-T)}$ has been added to the replicating portfolio.) The solution to this option-pricing problem is well known and first appeared in \citet{Black/Scholes}, equation~(13), and today can be found in many sources. We present it here to test our three-step process on a simple example. The solution to~\eqref{5.1} with terminal value~\eqref{5.2} was shown by \citet{Black/Scholes} to be
\begin{equation}\label{6.2}
\V(x, t) = N(z_0)x-N(z_1)Ke^{r(t-T)},
\end{equation}
for $(x,t)\in\R^{+}\times[0,T)$, where $N$ is the normal cumulative distribution function with
\begin{equation}\label{6.4}
z_0=\frac{1}{\s\sqrt{T-t}} \Big(\log \big(x/K e^{r(t-T)}\big) + \s^2\big(T-t\big)/2\Big)\quad\text{  and }\quad
z_1=z_0-\s\sqrt{T-t}.
\end{equation}
In this case neither the terminal value function $f$ nor the option-pricing function $\V$ is homogeneous as a function of $x\in\R^+$. 

Since the terminal value function $f$ in \eqref{5.2} is not strictly positive, to proceed we change it to
\begin{equation}\label{10.1}
\fh(x)=f(x)+K= (x-K)^++K=x\lor K>0,
\end{equation}
for $x\in\R^+$. With this terminal value function we now proceed with the three steps given above. Note that in the \citet{Black/Scholes} setting,~\eqref{repgen} becomes~\eqref{5.1}.

\vspace{10pt}
\noindent {\em Step~\ref{s1}:} Since the price process $X_0^{-1}X$ is of the same form~\eqref{11} as $X$, we can use the same  equation~\eqref{5.1} with $r=0$. Hence, we must solve
\begin{equation}\label{5.3}
\dd{\V}{ t}+\half\s^2 x^2 \dd{^2 \V}{x^2}=0,
\end{equation}
for $(x,t)\in\R^{+}\times[0,T)$, with  $\V(x,T)=\fh(x)$. From~\eqref{6.2} and~\eqref{6.4} we see that  the solution to this equation for $(x,t)\in\R^+\times[0,T)$ is 
\begin{equation}\label{5.3a}
\V(x,t)=N(z_0)x+(1-N(z_1))K,
\end{equation}
 with
\begin{equation} \label{10.3}
z_0=\frac{1}{\s\sqrt{T-t}} \Big(\log \big(x/K\big) + \s^2\big(T-t\big)/2\Big)\quad\text{  and }\quad
z_1=z_0-\s\sqrt{T-t}.
\end{equation}
The function $\V\colon\R^+\times[0,T)$ is $C^{2,1}$ and satisfies~\eqref{5.3}, and  $xD_1\V(x,t))/\V(x,t)$ is bounded for $(x,t)\in\R^+\times[0,T)$ since 
\[
D_1\V(x,t)= N(z_0),
\]
with $z_0$ as in~\eqref{10.3} (see \citet{Black/Scholes}, equation~(14)).
 Hence, $\V$ is a portfolio generating function that is replicable for the price processes $1, X$, where $X$ has constant variance rate $\s^2$.

\vspace{10pt}
\noindent  {\em Step~\ref{s2}:} We construct the contingent claim function $\Vh\colon\R^+\times\R^+\times[0,T]\to\R^+$ defined by
\begin{equation} \label{10.4}
\Vh(x_0,x,t)= x_0\V(x_0^{-1}x,t)=N(\zh_0)x+(1-N(\zh_1))Kx_0,
\end{equation}
for $(x_0,x,t)\in\R^+\times\R^+\times[0,T)$, with 
\begin{equation} \label{10.5}
\zh_0=\frac{1}{\s\sqrt{T-t}} \Big(\log \big(x/Kx_0\big) + \s^2\big(T-t\big)/2\Big)\quad\text{  and }\quad
\zh_1=\zh_0-\s\sqrt{T-t},
\end{equation}
and $\Vh(x_0,x,T)=x_0\fh(x_0^{-1}x)$.

\vspace{10pt}
\noindent  {\em Step~\ref{s3}:} We now have a contingent claim function $\Vh$ that solves the option-pricing problem
with the terminal value function $\fh$ of~\eqref{10.1} and is replicable for $X_0,X$. Here the only restriction on the riskless asset $X_0$ is that $X_0(T)=1$, whereas with the original solution, \eqref{6.2} and~\eqref{6.4}, the riskless asset was restricted to  $X_0(t)=e^{r(t-T)}$, for $t\in[0,T]$.

\vspace{10pt}
To solve the original problem with terminal value function $f$ of~\eqref{5.2}, we can subtract $Kx_0$ from~\eqref{10.4} since it  is also a solution for~\eqref{5.3}, which is linear and homogeneous.  Hence, the option-pricing problem with terminal value $f$ of~\eqref{5.2} is solved by
\begin{equation}\label{10.2}
\Vh(x_0,x,t)-Kx_0,
\end{equation}
for $(x_0,x,t)\in\R^+\times\R^+\times[0,T]$. Although we can recover the option price in this manner, the value of the hedging portfolio at $t\in[0,T]$ will be given by $Z_\p(t)-KX_0(t)$, which is incompatible with our logarithmic representation~\eqref{1.2} since its value can reach 0 at $t=T$. While $Z_\p(t)-KX_0(t)$ cannot be represented logarithmically on $[0,T]$,  the values  $Z_\p(t)$ and $X_0(t)$ can be represented, so the difference can be calculated, although not as a single portfolio value. It is possible that the arithmetic approach to generating functions of~\citet{Karatzas/Ruf:2017} could be used directly in this case, without adding and subtracting $Kx_0$ as in~\eqref{10.1} and~\eqref{10.2}. 
 
The Black-Scholes model has been expanded to include multiple correlated assets, as in \citet{Carmona2006} and \citet{Guillaume2019}, and solutions have been found for many other similar option-pricing problems, as in, e.g., \citet{Smith:1976} and  \citet{HN:2001,HN2:2001}. We shall not review these results here, but rather as a final example we consider the option-pricing problem with a terminal value function that is clearly inhomogeneous.

\begin{exa}\label{XA3} {\em Option pricing with an inhomogeneous terminal value function.} Suppose that  $n\ge2$ and $X_1,\ldots,X_n$ are price processes that are given by
\begin{equation}\label{A3.1}
d\log X_i(t)=\g_i(t)dt+\sum_{\ell=1}^{d}\zeta_{i,\ell}dW_\ell(t),
\end{equation}
for $i=1,\ldots,n$ and $t\in[0,T]$, as in~\eqref{1.0} with constant $\zeta_{i,\ell}$, so the covariance rate processes $\s_{ij}$ will also be constant. We want to find a contingent claim function with terminal value function
 \begin{equation}\label{A.1q}
f(x)=x_1^{p_1}+\cdots+x_n^{p_n},
 \end{equation}
for $x\in\R^{+n}$, where $p_i\ne p_j\in\R$, for $1\le i\ne j\le n$.  We approach this by  following the three steps above.

As Step~1, we must solve~\eqref{repgen} with $\g_0\equiv0$ and price processes  $1,X_0^{-1}X_1,\ldots,X_0^{-1}X_n$. Since the price processes $X_0^{-1}X_i$ have the same generic form \eqref{A3.1} as the $X_i$ we can instead use $1,X_1,\ldots,X_n$ in~\eqref{repgen}, so we must solve
\begin{equation}\label{A.1r}
\dd{\V}{t}+\half\sumi \s_i^2x^2_i\dd{^2\V}{x^2_i}=0,
\end{equation}
for $(x,t)\in\R^{+n}\times[0,T]$, since  $D_{ij}\V$ vanishes for $i\ne j$. We must solve this  with terminal value $\V(x,T)=f(x)$ for $x\in\R^{+n}$. Let us  find a portfolio generating function $\V\colon\R^{+n}\times[0,T]\to\R^+$ of the form
\begin{equation}\label{A.11x}
\V(x,t)=\sumi\ph_i(t) x_i^{p_i},
\end{equation}
for $(x,t)\in\R^{+n}\times[0,T]$, where the $\ph_i$ are positive $C^2$ functions with $\ph_i(T)=1$, so that $\V(x,T)=f(x)$ for $x\in\R^{+n}$. 

For $\V$ of the form~\eqref{A.11x}, the differential equation~\eqref{A.1r} becomes
\begin{equation}\label{A.1.1}
\sumi \ph'_i(t)x_i^{p_i}+\half\sumi (p_i^2-p_i)\s_i^2\ph_i(t)x_i^{p_i}=0,
\end{equation}
for   $(x,t)\in\R^{+n}\times[0,T]$, which will be solved if
\begin{equation}\label{A.1.2}
\ph'_i(t)+(p_i^2-p_i)\s_i^2\ph_i(t)/2=0,
\end{equation}
for  $i=1,\ldots,n$ and $t\in[0,T]$. A solution of~\eqref{A.1.2} will be of the form $\ph_i(t)=c_i e^{\a_it}$, for real constants $\a_i$ and $c_i$, so $\ph'_i(t)=\a_i\ph_i(t)$, and we have
\begin{equation}\label{A.1.3}
\a_i=(p_i-p_i^2)\s_i^2/2,
\end{equation}
for $i=1,\ldots,n$. Hence, a solution of~\eqref{A.1.1} will be of the form
\begin{equation}\label{A.1.4}
\V(x,t)=\sumi c_i e^{(p_i-p_i^2)\s_i^2t/2}x_i^{p_i},
\end{equation}
for  $(x,t)\in\R^{+n}\times[0,T]$, and if we let $c_i=e^{-(p_i-p_i^2)\s_i^2T/2}$, for $i=1,\ldots,n$, then
\begin{equation}\label{A.1.5}
\V(x,t)=\sumi e^{(p_i-p_i^2)\s_i^2(t-T)/2}x_i^{p_i},
\end{equation}
for   $(x,t)\in\R^{+n}\times[0,T]$, and this satisfies~\eqref{A.1r} with the terminal value $\V(x,T)=f(x)$ for $x\in\R^{+n}$.

To solve this option-pricing problem for $X_0,X_1,\ldots,X_n$ with a riskless asset $X_0$ for which $X_0(T)=1$,  we follow Steps~\ref{s2} and~\ref{s3} above, and we have
\[
\Vh(x_0,x,t)=x_0\V(x_0^{-1}x,t)=\sumi e^{(p_i-p_i^2)\s_i^2(t-T)/2}x_0^{(1-p_i)}x_i^{p_i},
\]
for $(x_0,x,t)\in\R^+\times\R^{+n}\times[0,T]$.  We see that $\Vh$ is homogeneous of degree~1, so it is a contingent claim function that satisfies~\eqref{A.1r}.

To calculate the weights of the portfolio $\p$ generated by $\Vh$, we can  apply Proposition~\ref{C0} and we see that  
\[
\p_i(t)=\frac{X_i(t)D_i\Vh(X_0(t),X(t),t)}{\Vh(X_0(t),X(t),t)}=\frac{p_ie^{(p_i-p_i^2)\s_i^2(t-T)/2}X_0^{(1-p_i)}(t)X_i^{p_i}(t)}{\Vh(X_0(t),X(t),t)},
\]
for $i=1,\ldots,n$ and $t\in[0,T]$, with
\[
\p_0(t)=\frac{X_0(t)D_0\Vh(X_0(t),X(t),t)}{\Vh(X_0(t),X(t),t)}=1-\sumi \p_i(t),
\]
for $t\in[0,T]$. Since $\Vh$ is replicable for $X_0,X_1,\ldots,X_n$, the portfolio $\p$ with initial value $Z_\p(0)=\Vh(X_0(0),X(0),0)$  satisfies
\begin{equation}\label{A.1.6}
Z_\p(t)=\Vh(X_0(t),X(t),t)=\sumi e^{(p_i-p_i^2)\s_i^2(t-T)/2}X_0^{(1-p_i)}(t)X_i^{p_i}(t),
\end{equation}
for $t\in[0,T]$. \qed
\end{exa}

\vspace{10pt}
\noindent{\bf Acknowledgements.} The authors thank Ioannis Karatzas and Martin Schweizer for their invaluable comments and suggestions, and the second author thanks the Institute for Advanced Study for providing the inspiration to pursue this research.

\appendix
\section{Appendix}\label{appendix}

\noindent{\bf Example \ref{X4a}, continued.} {\em The derivation of (\ref{4.9b}).} For constants $\s_1,\ldots,\s_n>0$, we wish to solve the differential equation 
\begin{equation}\label{a1}
\dd{}{t}\V(x,t)+\half\sumi \s_i^2 x^2_i \dd{^2}{x_i^2}\V(x,t)=0,
\end{equation}
for $(x,t)\in\R^{+n}\times[0,T]$,  with the terminal condition
\[
\V(x,T)=\big(x_1^{1/2}+\cdots+x_n^{1/2}\big)^2,
\]
for $x\in\R^{+n}$.  

We can change variables, with
\begin{equation}\label{A10}
y_i=\log x_i +\s_i^2 t/2,\qquad\text{ or }\qquad x_i=e^{y_i}e^{-\s_i^2 t/2},
\end{equation}
 for $t\in[0,T]$ and $i=1,\ldots,n$, so $y\in\R^n$ and  $U\colon\R^n\times[0,T]\to\R^+$ defined by
\[
U(y,t)=\V(x,t),
\]
for $(x,t)\in\R^{+n}\times[0,T]$, will be $C^{2,1}$. Hence, we have
\begin{align*}
D_t\V(x,t)&=\dd{}{t}U(y,t)\\
&=\half \sumi \s_i^2D_iU(y,t)+D_t U(y,t),
\end{align*}
for $(x,t)\in\R^{+n}\times[0,T]$, and since
\begin{align*}
D_i\V(x,t)&=\dd{}{x_i}U(y ,t)\\
&=x_i^{-1} D_iU(y,t),
\end{align*}
for $(x,t)\in\R^{+n}\times[0,T]$ and  $i=1,\ldots,n$, we also have
\begin{align*}
D_{ii}V(x,t)&= \dd{}{x_i} \Big( x_i^{-1} D_i U(y,t) \Big) \\
&=-x_i^{-2}D_iU(y,t)+x_i^{-2} D_{ii}U(y,t),
\end{align*}
for $(x,t)\in\R^{+n}\times[0,T]$ and  $i=1,\ldots,n$. Hence, we can now write \eqref{a1} as
\begin{equation*}
\dd{}{t}U(y,t)+\half\sumi \s_i^2\dd{^2}{y_i^2}U(y,t)=0,
\end{equation*}
for $(y,t)\in\R^n\times[0,T]$, and if we let
\[
\t=T-t\qquad\text{ and }\qquad \U(y,\t)=U(y,t)=\V(x,t),
\]
for $(y,\t)\in\R^n\times[0,T]$ and $x\in\R^{+n}$ given by~\eqref{A10}, then we have
\begin{equation}\label{4.11}
\dd{}{\t}\U(y,\t)-\half\sumi \s_i^2\dd{^2}{y_i^2}\U(y,\t)=0,
\end{equation}
for $(y,\t)\in\R^n\times[0,T]$, which is a heat equation (see \citet{Evans:2010}, Section~2.3).

We must solve~\eqref{4.11} with the initial condition
\[
\U(y,0)=\V(x,T)=\big(x_1^{1/2}+\cdots+x_n^{1/2}\big)^2=\big(e^{y_1/2}e^{-\s_1^2T/4}+\cdots+e^{y_n/2}e^{-\s_n^2T/4 }\big)^2,
\]
for $y\in\R^n$. The solution to this initial value problem is
\begin{align*}
\U(y,\t)&=\sumi \frac{e^{-\s_i^2T/2}}{(2\p\s_i^2\t)^{1/2}}\int_{\R}e^{-(y_i-z_i)^2/2\s_i^2\t}e^{z_i}dz_i\\
&\quad+ \sum^n_{\substack{i,j=1\\i\ne j}}\bigg(\frac{e^{-\s_i^2T/4}}{(2\p\s_i^2\t)^{1/2}}\int_{\R}e^{-(y_i-z_i)^2/2\s_i^2\t}e^{z_i/2}dz_i\bigg)
\bigg(\frac{e^{-\s_j^2T/4}}{(2\p\s_j^2\t)^{1/2}}\int_{\R}e^{-(y_j-z_j)^2/2\s_j^2\t}e^{z_j/2}dz_j  \bigg)\\
&= \sumi e^{-\s_i^2(T-\t)/2}e^{y_i}+\sum^n_{\substack{i,j=1\\i\ne j}} e^{-\s_i^2\t/8}e^{-\s_j^2\t/8}\big(e^{-\s_i^2(T-\t)/4}e^{y_i/2} \big)\big(e^{-\s_j^2(T-\t)/4}e^{y_j/2} \big),
\end{align*}
for $(y,\t)\in\R^n\times[0,T]$ (see \citet{Evans:2010}, Section~2.3, Theorem~1, or \citet{Duffie:2001}, Chapter~5, Section~I). Hence,
\begin{equation*}
\V(x,t)= \sumi x_i+\sum^n_{\substack{i,j=1\\ i\ne j}}e^{\s_i^2(t-T)/8} e^{\s_j^2(t-T)/8}x_i^{1/2}x_j^{1/2},
\end{equation*}
for $(x,t)\in\R^{+n}\times[0,T]$.
 \qed

\bibliographystyle{chicago}
\bibliography{math,math1,math2,math4}

\end{document}